\documentclass[conference,letterpaper]{IEEEtran}
\IEEEoverridecommandlockouts

\def\ba{{\mathbf{a}}}

 \def\bv{{\mathbf{v}}} \def\bw{{\mathbf{w}}} \def\bx{{\mathbf{x}}}
\def\by{{\mathbf{y}}}

\def\bU{{\mathbf{U}}}   \def\bX{{\mathbf{X}}}
\def\bY{{\mathbf{Y}}}

\newcommand{\E}{\mathbb{E}}

\usepackage{bbm}
\usepackage{amsmath}
\usepackage{tcolorbox}
\usepackage{breqn}
\usepackage{cite}
\usepackage{amssymb}
\usepackage{mathtools}
\usepackage{amsthm}
\newtheorem{theorem}{Theorem}

\newtheorem{lemma}{Lemma}
\theoremstyle{remark}
\newtheorem*{rem}{Remark}
\usepackage{accents}
\usepackage{xcolor}

\begin{document}

\title{Stochastic-Adversarial Channels : Online Adversaries  With Feedback Snooping
\thanks{This work was funded in part by the National Science Foundation under grants CNS1642982, CCF1816013, and EEC1941529. }
\thanks{This paper is an extended draft of the conference paper with the same title submitted to the IEEE International Symposium on Information Theory (ISIT) 2021. }
}

\author{\IEEEauthorblockN{Vinayak Suresh, Eric Ruzomberka and David J. Love}
\IEEEauthorblockA{\textit{School of Electrical and Computer Engineering} \\
\textit{Purdue University}\\
Email: suresh20@purdue.edu, eruzombe@purdue.edu, djlove@purdue.edu}
}

\maketitle

\begin{abstract}
The growing need for reliable communication over untrusted networks has caused a renewed interest in adversarial channel models, which often behave much differently than traditional stochastic channel models. Of particular practical use is the assumption of a \textit{causal} or \textit{online} adversary who is limited to causal knowledge of the transmitted codeword. In this work, we consider stochastic-adversarial mixed noise models. In the set-up considered, a transmit node (Alice) attempts to communicate with a receive node (Bob) over a binary erasure channel (BEC)  or binary symmetric channel (BSC) in the presence of an online adversary (Calvin) who can erase or flip up to a certain number of bits at the input of the channel. Calvin knows the encoding scheme and has causal access to Bob's reception through \textit{feedback snooping}. For erasures, we provide a complete capacity characterization with and without transmitter feedback. For bit-flips, we provide interesting converse and achievability bounds. 

\end{abstract}

\section{Introduction}

A central endeavour in information theory is the study of capacity limits and coding strategies for reliable communication over different types of channels. Two different philosophies exist on how channels are modeled. Channels in the Shannon world are characterized by some stochastic process that injects errors independently of the communication scheme, while channels in the Hamming world are characterized by an adversary who injects worst-case errors. Historically, adversarial channels were studied under either full knowledge (\textit{omniscient adversary}) or no knowledge (\textit{oblivious adversary}) of the transmitted codeword. A number of recent works \cite{langberg2009binary,code_online_adv_old,improved_upbounds,dey2013upper,bassily2014causal,chen2015characterization,error_erasure,online_largea} instead consider coding against \textit{online} or \textit{causal} adversaries wherein at any point during the transmission, the adversary knows only part of the codeword transmitted thus far. 

As noted in \cite{dey2013upper}, the causal adversary model lies in between the stochastic and the omniscient adversary models. In this work, we further bridge together the Shannon and the Hamming worlds by studying a new model where both adversarial and random sources of error are present. Specifically, Alice attempts to send a message to Bob over a binary erasure channel BEC($q$)  or binary symmetric channel BSC($q$) in the presence of a causal adversary Calvin who can erase or flip a certain number of bits at the input of the channel. This is depicted in Fig \ref{channelmodel}. Any transmission strategy must not only overcome the noise due to the random channel but also from the adversary. We also assume that Calvin has access to Bob's reception, which we refer to  as \textit{feedback snooping}. The ability to spy on both Alice and Bob aids Calvin in designing strong attacks. Our goal is to characterize the capacity of this channel.

When there is no random channel present, i.e., $q=0$ in Fig. \ref{channelmodel}, the only source of noise is adversarial. A complete capacity characterization for this case is given in \cite{dey2013upper,chen2015characterization,bassily2014causal}. Our models differ from the ones considered previously in two ways:
\begin{itemize}
    \item \textbf{Mixture of random and adversarial noise} - The noise in the received word is affected by the random channel BEC (BSC) as well as the actions of Calvin who is erasing (flipping) bits. For example in the erasure case, a bit not erased by Calvin can be erased by the BEC. Similarly, in the bit-flip case, a bit flipped by Calvin may be ``unflipped" by the BSC. Conceptually, we think of the discrete memoryless channel (DMC) as the main channel through which Alice and Bob communicate, and Calvin as a malicious entity who attempts to disrupt the transmission.

    \item \textbf{Feedback to adversary} - In our setting, Calvin is allowed access to Bob's reception through \textit{feedback snooping}. This becomes important due to the presence of the stochastic channel. The adversarial attacks described in \cite{dey2013upper,bassily2014causal} if used directly do not provide the right distance bounds needed to establish our converse results. These are appropriately strengthened and crucially rely on Calvin's ability to snoop. Note that feedback snooping is unnecessary when $q=0$.
    
\end{itemize}
Our contributions can be summarized as follows:
\begin{itemize}
    \item We provide a complete characterization of capacity for the case of erasures. Our result implies that the presence of the random channel BEC($q$) in addition to adversarial erasures simply scales the capacity by a multiplicative factor. 
    \item For the case of erasures, we also characterize the capacity when Alice has causal access to Bob's reception and encoding is \textit{closed-loop}. In this scenario, we show that Calvin gains no benefit from his ability to spy on Alice or Bob.  In fact,  he can do no better than making erasure decisions in an i.i.d. manner.
     \item Finally in the case of bit-flips, we prove non-trivial converse and achievability bounds. 
\end{itemize}

There are other adversarial models intermediate between the oblivious and omniscient models that have been considered in literature which we do not pursue here (e.g.,  \cite{sufficientlymyopic,dey2010codingdelayed,dey2016bitofdelay,dey2019interplay,zhang2018covert,amita1}). The problem of coding with feedback to the transmitter has been studied by several authors such as \cite{berlekamp1964block,lebedev2016coding,zigangirov1976number}. Finally, we note that our models can be cast under the more general framework of arbitrarily varying channels (AVCs) \cite{csiszarkornerbook,surveyAVC}.  However, known results for AVCs do not directly imply the results in this paper. 

\section{Channel Models}

Consider the channel depicted in Fig.\ref{channelmodel}. Alice (the transmitter) attempts to convey a message to Bob (the receiver) over a BEC($q$), in the presence of a $p$-limited causal adversary (Calvin) where the terms will be clarified shortly. The input and output alphabets are $\mathcal{X} = \lbrace 0,1 \rbrace$ and $\mathcal{Y} = \lbrace 0,1,\Lambda\rbrace, $ respectively, where $\Lambda$ denotes an erasure symbol. Encoding is done over $n$ channel uses, and the size of the message set at the transmitter is $2^{nR}$. We allow stochastic encoding and assume the presence of local randomness available only to Alice for this purpose. Denote $x_k \in \mathcal{X}$ to be the symbol selected by the transmitter at channel use $k$. At time $k$, the adversary makes a decision on whether to erase $x_k$  based on his side-information to be specified later. If Calvin erases $x_k$, the received symbol at time $k$ at the receiver is an erasure, i.e., $y_k = \Lambda$. If Calvin decides not to erase $x_k$, then $y_k = x_k$ with probability $1-q$ and $y_k = \Lambda$ with probability $q$, i.e., $x_k$ is erased with probability $q$. 

We assume that Calvin knows the codebook used at the transmitter in the case of deterministic encoding or the distribution of codewords in the case of stochastic encoding. Calvin is assumed to be \textit{causal}, i.e., at each channel use $k$, he knows only part of the codeword transmitted so far $(x_1,x_2,\cdots,x_{k}) \in \mathcal{X}^{k}$. Calvin is neither aware of the message nor future transmissions. However, he has access to Bob's reception $(y_1,y_2,\cdots,y_{k-1}) \in \mathcal{Y}^{k-1}$ through a delay-free and noise-free causal feedback link as shown in Fig. \ref{channelmodel}. 

A power constraint is further imposed by enforcing Calvin to be $p$-limited, meaning that he can erase at most a constant fraction $p$ of the bits, i.e., if $\ba \in \{ 0,\Lambda \}^n$ denotes the positions where Calvin decides to erase symbols from $(x_1,x_2,\cdots,x_{n})$, we must have $weight(\mathbf{a}) \leq pn$. We refer to this model as \textit{the BEC causal adversarial channel with feedback snooping} (or BEC($q$)-ADV($p$)-FS). Note that the BEC block in Fig. \ref{channelmodel} is slightly different from the classical BEC. If Calvin erases $x_k$ to an erasure symbol $\Lambda$, we have $y_k = \Lambda$, where $\Lambda$ does not carry any information.

\begin{figure}
    \centering
    \includegraphics[width=0.47\textwidth]{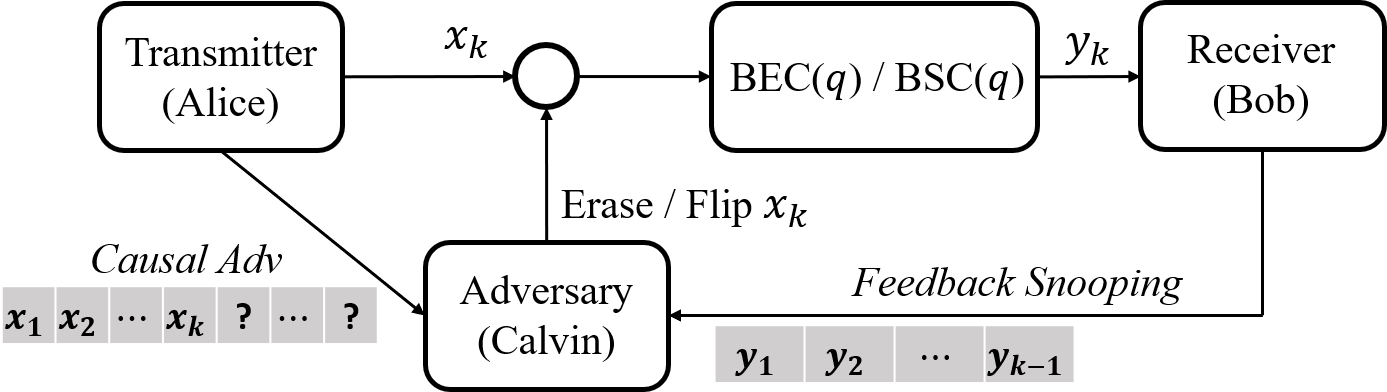}
    \caption{Channel models considered in this work. }
    \label{channelmodel}
\end{figure}

Our aim is to characterize the capacity of this channel, i.e., the largest value of $R$ such that Alice can reliably convey one out of $2^{nR}$ possible messages to Bob. Precise definitions are given shortly. In Section \ref{BFchannel}, we also consider a related channel by replacing the BEC($q$) with a BSC($q$) and letting Calvin flip bits instead of erasing them, denoted henceforth as BSC($q$)-ADV($p$)-FS.

\textit{Notation and Definitions:} In this work, we only consider fixed length encoding. The blocklength is denoted by $n$. The transmitted message is denoted by the random variable (r.v.) $\bU$ chosen uniformly from the message set $\mathcal{U}=\lbrace 1,2,3,\cdots,2^{nR} \rbrace$. A deterministic code consists of a fixed encoder map $\Phi_d:\mathcal{U} \to \mathcal{X}^n$ and a decoder map $\Gamma_d:\mathcal{Y}^n\to \mathcal{U}$, where each message is associated to a unique codeword. In case of stochastic encoding, a codeword $\bx$ is selected for a message $u$ according to a chosen conditional distribution $\Phi(.|u)$ defined on $\mathcal{X}^n$. A stochastic code is fully specified by defining all conditional distributions $\left \lbrace \Phi(.|u) \right \rbrace_{u\in \mathcal{U}}$ and a decoder $\Gamma : \mathcal{Y}^n \to \mathcal{U}$. Without loss of generality, we assume in proving converse results that no two distinct messages map to the same codeword. The (maximum) probability of error is then 
\begin{equation}\label{maxerrorprob}
    P_e = \max_{u \in \mathcal{U}}\max_{\text{ADV($p$)}}\sum_{\by}\sum_{\bx} P(\by | \bx) \Phi(\bx|u) \mathbbm{1}(\Gamma(\by)\neq u)
\end{equation}
where $\mathbbm{1}(.)$ denotes the indicator function and ADV($p$) denotes a feasible strategy chosen by Calvin. Note that $P(\by|\bx)$ in \eqref{maxerrorprob} is a function of both the stochastic channel and the chosen adversarial strategy. We say that $R>0$ is achievable if for every $\delta>0$ and every sufficiently large $n$, there is a code of rate $R$ and blocklength $n$ with $P_e < \delta$. The capacity is defined to be the supremum of all achievable rates. Let Ber($q$) denote a Bernoulli r.v. with success probability $q$. For $x,y\in[0,1/2]$, let $x \star y = x(1-y)+y(1-x)$ and note that $x \star y = 1/2$ iff either $x=1/2$ or $y=1/2$ (or both).

\section{Results for Erasures}

\subsection{No Transmitter Feedback}
Denote by $C^E(p,q)$ the capacity of BEC($q$)-ADV($p$)-FS when Alice has no side-information, i.e., encoding is restricted to be \textit{open-loop}. We prove the following result.

\begin{theorem}\label{thm_er}
The capacity $C^E(p,q)$ of BEC($q$)-ADV($p$)-FS is given by
\begin{equation}\label{drloveresult_er}
    C^E(p,q) =     \begin{cases}
    (1-2p)(1-q) & \text{ for $0\leq p\leq \frac{1}{2}, ~ 0 \leq q \leq 1$} \\
    0 & \text{otherwise}
    \end{cases}.
\end{equation}
\end{theorem}
\begin{rem}
When there is no BEC, i.e., when $q=0$, our model reduces to the one studied in \cite{dey2013upper, chen2015characterization}. Our result implies that in the setting where both causal adversarial erasures and random erasures are present, the capacity simply scales by a factor of $1-q$. 
\end{rem}

\begin{proof}

\textbf{Converse :}
The proof of converse is based on a \textit{wait and snoop, then push} attack inspired by, but different from, an attack in \cite{langberg2009binary,bassily2014causal}. Let the transmitted and the received codewords be denoted by $\bx$ and $\by$ respectively. Let $\bx_1=(x_1,x_2,\cdots,x_\ell)$ and $\bx_2=(x_{\ell+1},\cdots,x_n)$, where $\ell$ is specified shortly. Similarly, let $\by_1=(y_1,y_2,\cdots,y_\ell)$ and $\by_2=(y_{\ell+1},\cdots,y_n)$. 

Suppose Alice attempts to communicate at a rate $R= C^E(p,q)+\epsilon=(1-2p)(1-q) + \epsilon$. We will show that for sufficiently large block-length $n$, the probability of decoding error under the proposed attack is lower bounded by a constant that is only a function of $\epsilon$ (and independent of $n$). The two phases of the attack are: 
\begin{itemize}
    \item \textbf{Wait and Snoop}: Calvin waits and does not induce any erasures for the first $\ell = n\frac{R-\frac{\epsilon}{2}}{1-q}$ channel uses. Instead, Calvin simply snoops into Bob's reception to determine the erased/unerased bits and their positions. At the end of this phase, Bob receives $\by_1=(y_1,y_2,\cdots,y_{\ell})$ containing some erased and some unerased bits. Note that the erasures in this phase occur purely due to the BEC($q$) channel. Let $\lbrace i_j \rbrace_{j=1}^m$ be the indices of symbols in $\by_1$ that remain unerased. Here, $m$ is a random quantity in accordance to the erasure distribution from the BEC($q$). 
    \item \textbf{Push}: Calvin forms the set $\mathcal{B}_{\by_1}$ of codewords consistent with $\by_1$ as 
\begin{multline}\label{Econfset}
    \mathcal{B}_{\by_1} = \lbrace \bv \in \mathcal{X}^n : \exists \tilde{u} \in \mathcal{U} \text{ s.t. } \Phi(\bv|\tilde{u})>0 \text{ and } \\  v_{i_k} = x_{i_k} ~ k = 1,2,\cdots, m  \rbrace,
\end{multline}
where $\Phi(.|u)$ is the distribution of codewords when message $u$ is to be transmitted. In other words, $\mathcal{T}_{\by_1}$ consists of all possible codewords that align with $\by_1$ at the positions that are unerased.
Calvin then samples a codeword $\bx'$ from $\mathcal{B}_{\by_1}$ according to the distribution $P_{\bX|\bY_1=\by_1}(.|\by_1)$. In the push phase, Calvin simply erases bit $x_i$ whenever $x_i \neq x_i'$. To complete the proof, it suffices to show that $\bx$ and $\bx'$ correspond to different messages $u$ and $u'$ and that $d(\bx_2,\bx_2')<pn$ with a probability \textit{independent} of $n$. This way there is no way for Bob to distinguish between messages $u$ and $u'$. This is illustrated in Fig. \ref{push_erasure_fig}.
\end{itemize}

\begin{figure}
    \centering
    \includegraphics[width=0.45\textwidth]{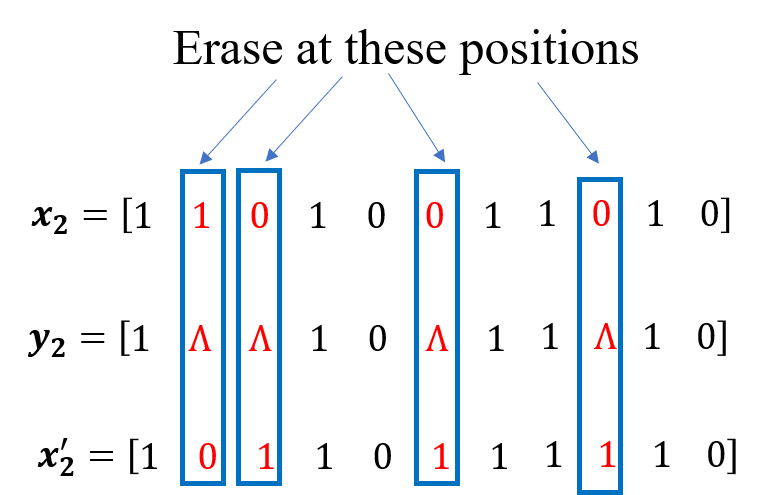}
    \caption{In the push phase, if $\bx_2$ and $\bx_2'$ are sufficiently close (within distance $pn$), Calvin can make Bob completely uncertain whether the transmitted codeword was $\bx$ or $\bx'$.}
    \label{push_erasure_fig}
\end{figure}

Note that while the presence of the BEC($q$) lowers the target rate, Calvin adds no erasures for approximately $n(1-2p)$ channel uses which from \cite{bassily2014causal,chen2015characterization} is optimal when there is no BEC($q$). The main difference in attack when $q\neq 0$ is that even though Calvin knows the entire prefix of the transmitted codeword $\bx_1 = (x_1,x_2,\cdots,x_{\ell})$, he forms his set in \eqref{Econfset} based only on the unerased bits. Thanks to feedback snooping, Calvin exploits the additional equivocation induced by the BEC($q$) in the wait and snoop phase to pick a codeword that is sufficiently close to the transmitted codeword, and which corresponds to a message different from one that Alice chose. Note also that while we give Calvin full causal access to Bob's reception, an alternate model where Calvin is allowed \textit{one-time block feedback} is sufficient - he would add no erasures for $\ell$ channel uses, retrieve through feedback the entire block $\by_1$ and then `push'.

The proof steps are similar to section A from \cite{bassily2014causal} except that we account for the presence of the BEC($q$) in our claims. Define the set $
    A_0 = \left \lbrace  \by_1 :  H(\bU \mid \bY_1 = \by_1) > \frac{n\epsilon}{4}  \right \rbrace
$ and the event $E_1=\lbrace \bY_1 \in A_0 \rbrace$. We have the following lemma.

\begin{lemma}
     $P(E_1 )\geq \frac{\epsilon}{4}$.
\end{lemma}
\begin{proof}
Since $\bU\to \bX_1 \to \bY_1$ is a Markov chain, by the data processing inequality, we have
\begin{equation*}
    I(\bU;\bY_1) \leq I(\bX_1,\bY_1) = \ell (1-q) = n(R - \epsilon /2).
\end{equation*}
The above holds since Calvin adds no erasures in the wait and snoop phase and the channel between $\bX_1$ and $\bY_1$ is a BEC($q$). Now, since $H(\bU) = nR$, we have
\begin{equation*}
    H(\bU|\bY_1) =  \mathbb{E}_{\bY_1} H(\bU|\bY_1 = \by_1) =  H(\bU) - I(\bU;\bY_1) \geq n\epsilon / 2.
\end{equation*}
By Markov's inequality then,
\begin{equation*}
    P\left ( nR - H(\bU|\bY_1 = \by_1) > nR - n\epsilon/4\right ) \leq 1 - \frac{\epsilon/4}{R - \epsilon/4}
\end{equation*}
which gives as desired,
\begin{equation*}
   P(E_1) = P\left ( H(\bU \mid \bY_1 = \by_1) > \frac{n\epsilon}{4}  \right ) \geq \frac{\epsilon}{4}.
\end{equation*}
\end{proof}

Now let $E_2$ be the event $\lbrace \bU \neq \bU'\rbrace $ and $E_3$ be the event $\lbrace d(\bX_2,\bX_2')<pn \rbrace$. First, we will show that for $\by_1 \in A_0$, $P(E_2,E_3 \mid \lbrace \bY_1=\by_1\rbrace) \geq \epsilon^{\mathcal{O}(1/\epsilon)}$. To that end, consider sampling $t=\frac{9}{\epsilon}$ codewords $\mathcal{C}_t = \left \lbrace \bX^{(1)},\bX^{(2)},\cdots,\bX^{(t)} \right \rbrace$ from the set $\mathcal{B}_{\by_1}$ where each codeword is sampled according to the conditional distribution $P_{\bX|\bY_1=\by_1}(.|\by_1)$. Let the messages corresponding to the codewords be $\bU_1,\bU_2,\cdots,\bU_t$ and let $E_4$ be the event that $\left \lbrace \bU_1, \bU_2, \cdots \bU_t \text{ are all distinct}\right \rbrace$ i.e. all of the codewords are distinct. We have from proposition 1, section A.2 from \cite{bassily2014causal} the following. 

\begin{lemma} \cite{bassily2014causal}\label{er_distlemma}
For $\by_1 \in A_0$ and for sufficiently large block length $n$,
$$P(E_4 \mid  \bY_1 = \by_1 ) \geq \left (\frac{\epsilon}{5}\right )^{t-1}.$$
\end{lemma}

The average Hamming distance between the suffixes of codewords in $\mathcal{C}_t$ is defined as 
\begin{equation*}
    d_{avg}(\mathcal{C}_t) = \frac{1}{t(t-1)} \sum_{i\neq j} d_H\left ( \bX_2^{(i)} ,\bX_2^{(j)}\right ).
\end{equation*}
Conditioning on $E_5$, Plotkin's bound dictates that
\begin{equation*}
\begin{split}
       d_{avg}(\mathcal{C}_t) \leq \frac{1}{2}\frac{t}{t-1}(n-\ell) &= n \frac{t}{t-1}\left ( p - \frac{\epsilon}{4(1-q)}\right ) \\
       &\leq  n \frac{\frac{9}{\epsilon}}{\frac{9}{\epsilon}-1}\left ( p - \frac{\epsilon}{4}\right ) \leq np - n\frac{\epsilon}{8}.
\end{split}
\end{equation*}
Thus for $\by_1 \in A_0$, we have
\begin{equation*}
    \mathbb{E}(d_{avg}(\mathcal{C}_t) \mid E_4, \bY_1 = \by_1 ) \leq np - n\epsilon/8.
\end{equation*}
Now, since all of the $\bX^{(i)}$'s are picked independently, all pairs $(\bX^{(i)}, \bX^{(j)})$ have identical distribution. Thus,
\begin{multline*}
   \mathbb{E}(d_{avg}(\mathcal{C}_t) \mid E_4, \bY_1 = \by_1 ) =  \\ \mathbb{E}(d_H(\bX_2^{(1)},\bX_2^{(2)}) \mid E_4, \bY_1 = \by_1 )    
\end{multline*}
and also
\begin{multline*}
 \mathbb{E}(d_H(\bX_2^{(1)},\bX_2^{(2)}) \mid E_4, \bY_1 = \by_1 ) = \\ \mathbb{E}(d_H(\bX_2,\bX_2') \mid E_4, \bY_1 = \by_1 ).
\end{multline*}
Thus, we have
\begin{equation*}
    \mathbb{E}(d_H(\bX_2^{(1)},\bX_2^{(2)}) \mid E_4, \bY_1 = \by_1 ) \leq  np - n\epsilon/8
\end{equation*}
and by Markov's inequality
\begin{equation}\label{er_markov}
    P(d_H(\bX_2^{(1)},\bX_2^{(2)})>np \mid E_4, \bY_1 = \by_1 ) \leq 1 - \frac{\epsilon}{8p}.
\end{equation}
We have also, 
\begin{multline*}
    P(E_2,E_3 \mid \bY_1 = \by_1 ) = \\ P(d(\bX_2^{(1)},\bX_2^{(2)}) \leq pn, \bU_1 \neq \bU_2 \mid  \bY_1 = \by_1 ) \\ \geq     P(d(\bX^{(1)}_2,\bX^{(2)}_2)  \leq pn, E_4 \mid \bY_1 = \by_1 ).
\end{multline*}
where the last inequality holds because event $E_4$ is a subset of the event $\lbrace U_1\neq U_2 \rbrace$. We then have,
\begin{multline*}
     P(d(\bX^{(1)}_2,\bX^{(2)}_2)\leq pn, E_4 \mid \bY_1 = \by_1 ) = \\
     P(d(\bX^{(1)}_2,\bX^{(2)}_2) \leq pn \mid E_4, \bY_1 = \by_1 ) P(E_4\mid \bY_1 = \by_1  ).
\end{multline*}
From Lemma \ref{er_distlemma} and \eqref{er_markov}, when $E_1$ occurs i.e. $\by_1 \in A_0$, we get,
\begin{equation*}
        P(E_2,E_3 \mid \bY_1 = \by_1 ) \geq \frac{\epsilon}{8p}\left (\frac{\epsilon}{5}\right )^{\frac{9}{\epsilon}-1} = \epsilon^{\mathcal{O}(1/\epsilon)}
\end{equation*}
as we set out to prove.

Recall that $E_2$ is the event that the message $\bU'$ picked by the adversary is different from the one transmitted and $E_3$ is the event that the corresponding codewords $\bX_2$ and $\bX_2'$ are close enough so that Calvin's push phase succeeds and Bob is completely uncertain whether the message transmitted was $\bU$ or $\bU'$. Hence when $E_2$ and $E_3$ occur, the probability of decoding error is at least $1/2$. To finish the proof, we need only show a lower bound on $P(E_2,E_3)$. We have,
\begin{equation*}
\begin{split}
P(E_2,E_3) &\geq P(E_2,E_3,E_1) \\
&= \sum_{\by_1 \in A_0}  P(E_2,E_3 \mid \bY_1 = \by_1 ) P(\bY_1=\by_1) \\
&\geq \frac{\epsilon}{8p}\left (\frac{\epsilon}{5}\right )^{\frac{9}{\epsilon}-1}\sum_{\by_1 \in A_0}P(\bY_1=\by_1) \\
&= \frac{\epsilon}{8p}\left (\frac{\epsilon}{5}\right )^{\frac{9}{\epsilon}-1} P(E_1)\\
&\geq \frac{\epsilon}{4} \frac{\epsilon}{8p}\left (\frac{\epsilon}{5}\right )^{\frac{9}{\epsilon}-1},
\end{split}
\end{equation*}
a lower bound that is independent of $n$, hence completing the proof.

\begin{figure}
    \centering
    \includegraphics[width=0.48\textwidth]{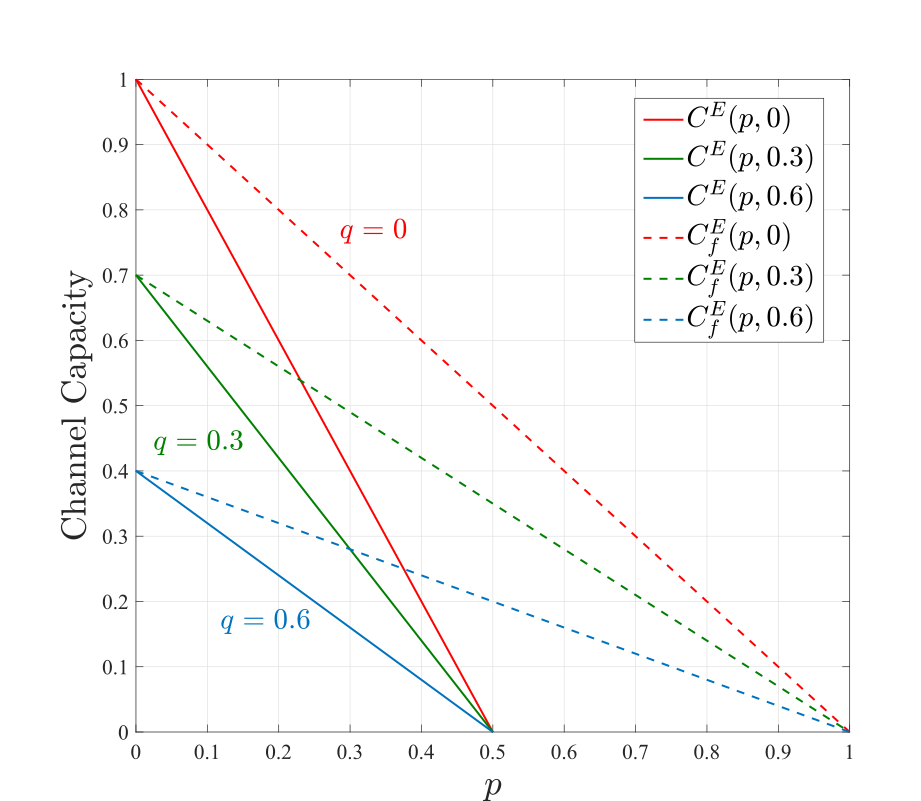}
    \caption{Capacity of BEC($q$)-ADV($p$)-FS with ($C_f^E(p,q))$ and without transmitter feedback (resp. $C^E(p,q))$ as a function of $p$ for $q=0,0.3,0.6$. The cut-off value of $p$ beyond which $C^{E}(p,q)=0$ is $p=1/2$ independent of $q$.}
    \label{BECupperbounds}
\end{figure}

\textbf{Achievability}: We resort to a random coding argument to claim existence of a stochastic code that achieves \eqref{drloveresult_er}. Let $R= C^E(p,q)-\epsilon=(1-2p)(1-q) - \epsilon$. Our construction is a modification of the encoder and decoder described in \cite{chen2015characterization}, which was used to prove (tight) achievability when $q=0$. We begin by briefly reviewing the encoder and decoder of \cite{chen2015characterization}. While reviewing, we provide key insights into how this decoder might fail once a BEC is added to the channel model. Following the review, we use our insights to modify the decoder in order to account for the additional random noise when $q>0$. Alice has a set of private secret keys $\mathcal{S}$ she uses for (stochastic) encoding. The encoder and decoder of \cite{chen2015characterization} is constructed as follows (here, $q=0$): 
\begin{itemize}
\item \textbf{Encoder:} A message $u$ is mapped to $1/\theta$ sub-codewords or chunks, each of length $n \theta$, where $\theta=\frac{\epsilon}{4}$ is a quantization parameter. Each of the sub-codewords is obtained by a stochastic code and the secrets between chunks are chosen independently. The sub-codewords are then concatenated together to form the transmitted word. Further technical details are in the appendix.

 \item \textbf{Decoder:} Decoding begins after Bob receives the entire $n$-symbol channel output $\mathbf{y}$. For some integer $t^*$, Bob partitions $\mathbf{y}$ into 2 strings: $\mathbf{y}_{1} = (y_1, \ldots, y_{t^*})$ and $\mathbf{y}_{2}= (y_{t^*+1}, \ldots, n)$. Decoding occurs in two sequential phases. In the first phase, Bob performs list decoding on $\mathbf{y}_1$ to create a list of messages $\mathcal{L}$. In the second phase, he refines the list by removing all messages in $\mathcal{L}$ that are not consistent with $\mathbf{y}_2$. Here, a message $u'$ is said to be consistent with $\mathbf{y}_2$ iff some codeword corresponding to $u'$ agrees with $\mathbf{y}_2$ on the unerased bits. If exactly one message, say $\hat{u}$, remains in $\mathcal{L}$ after refinement, the decoder outputs $\hat{u}$. If the refined list does not contain exactly one message, a decoding error is declared. Decoding is successful if $\hat{u} = u$.

\end{itemize}
The integer $t^*$ is a decoding point which indicates which part of the channel output is used for list decoding (phase 1) and which part is used for list refinement (phase 2). The authors in \cite{chen2015characterization} show that Bob can choose $t^*$ such that decoding is successful w.h.p.. Here, $t^*$ is chosen as a function of the number of (purely adversarial) erasures $\lambda^{a}_{t^*}$ observed in $\mathbf{y}$ up until time $t^*$. Specifically, Bob chooses $t^*$ as the smallest integer that satisfies the so-called \textit{list-decoding condition}
\begin{equation} \label{eq:phase1_cond}
\lambda^{a}_{t^*} \leq t^*(1-\theta) - ((1-2p)-\epsilon)n
\end{equation}
and the \textit{energy bounding condition} 
\begin{equation} \label{eq:phase2_cond}
np - \lambda^{a}_{t^*} \leq \frac{(n-t^*)(1-\theta)}{2}.
\end{equation}
Condition (\ref{eq:phase1_cond}) ensures the size of $\mathcal{L}$ is small (at most a constant) while condition (\ref{eq:phase2_cond}) ensures the fraction of erasures that occur in $\mathbf{y}_2$ is small enough to perform list refinement. 

Problems in this construction arise when $q>0$. If the decoder assumes that all erasures that he sees are adversarial and performs decoding by selecting $t^*$ according to conditions \eqref{eq:phase1_cond} and \eqref{eq:phase2_cond}, the maximum rate that can be achieved is $C^E(p+q-pq,0) = C^E(p,q)-q$ which is strictly less than capacity. Therefore, simply counting erasures without knowing (or estimating) their source is no longer a viable strategy when $q>0$. 

To circumvent the issues described above, we modify conditions \eqref{eq:phase1_cond} and \eqref{eq:phase2_cond} appropriately. Let $\lambda_t$ denote the number of erasures observed by Bob up until time $t$, which includes contributions both from Calvin and the BEC($q$). Then, Bob chooses $t^*$ as the smallest integer that satisfies the modified list-decoding condition
\begin{equation} \label{eq:phase1_cond_m}
{\lambda}_{t^*} - qt^*  \leq t^*(1-q)(1-\theta) - Rn 
\end{equation}
and the modified list refinement condition
\begin{equation} \label{eq:phase2_cond_m}
np(1-q) - ({\lambda}_{t^*}- qt^*)  \leq \frac{(n-t^*)(1-q)(1-\theta)}{2}.
\end{equation}
Note that if Calvin adds $\lambda^{a}_{t^*}$ erasures up until $t^*$, the total number of erasures  ${\lambda}_{t^*}$ that Bob observes is approximately ${\lambda}_{t^*} \approx \lambda^{a}_{t^*} + q(t^* - \lambda^{a}_{t^*})$. On making this substitution we see that $t^*$ satisfying \eqref{eq:phase1_cond_m} and \eqref{eq:phase2_cond_m} is nearly the same as that satisfying \eqref{eq:phase1_cond} and \eqref{eq:phase2_cond} i.e. it is sufficient to choose $t^*$ only as a function of pure adversarial erasures. However, since Bob has no way to ascertain this, he works with the quantity ${\lambda}_{t^*} - qt^*$ which is an estimate of the number of adversarial erasures that do not conincide with random erasures. Having selected $t^*$, Bob can then finish decoding using the two-phase decoding process of \cite{chen2015characterization} to successfully recover the transmitted message. Further details of the proof are provided in the appendix.

\end{proof}

\subsection{With Transmitter Feedback}
Suppose now that Alice in addition to Calvin has access to Bob's reception perfectly through a separate causal feedback link. This allows Alice to employ \textit{closed-loop} encoding strategies where the input $x_k$ at time $k$ is possibly a function of both the message and Bob's reception thus far $(y_1,y_2,\cdots,y_{k-1})$, i.e.,
\begin{equation}
    \bX_k \sim f_k(\bU,\bY_1,\bY_2,\cdots,\bY_{k-1}) ~~ k=1,2,\cdots,n
\end{equation}
where for each $k$, $f_k$ is either deterministic or, more generally, a probabilistic map defining a conditional distribution $P_{\bX |\bU,\bY_1,\bY_2,\cdots,\bY_{k-1}}$ over $\mathcal{X}$. Calvin is assumed to be causal. He does not know the message but knows the closed-loop encoding (possibly stochastic) maps $\lbrace f_k \rbrace_{k=1}^{n}$ used by Alice. Let the capacity in this case be denoted as $C_f^{E}(p,q)$. We have the following result.

\begin{theorem}
\label{causalBECtcfbk}
The capacity $C_f^E(p,q)$ of BEC($q$)-ADV($p$)-FS with causal feedback to the transmitter is
\begin{equation}\label{causalBECtxfbk_formula}
    C_f^E(p,q) =(1-p)(1-q) \quad \forall ~ 0\leq p \leq 1, 0 \leq q \leq 1 .
\end{equation}
\end{theorem}
\begin{rem}
If Calvin were to simply erase each symbol with probability $p$, the effective channel is a BEC with erasure probability $s=1-(1-p)(1-q)$. This means that the rate is limited to\footnote{For a vanilla DMC such as the BEC, the capacity is the same under deterministic and stochastic encoding \cite{surveyAVC}.} $(1-p)(1-q)$ which matches with the expression in \eqref{causalBECtxfbk_formula}. This implies that \textit{the optimal attack for the adversary is to simply cause i.i.d. erasures. The knowledge of the (closed-loop) encoding scheme or the ability to snoop into Bob's reception does not buy Calvin any benefit.} 
\end{rem}
\begin{proof}
\textbf{Converse:} The converse proof follows the above remark. Fix $\epsilon>0$. Calvin simply erases each symbol with probability $p-\frac{\epsilon}{1-q}$. By the Chernoff bound, the probability that Calvin will run out of his budget of $pn$ erasures decays exponentially with $n$. The combined effect of the adversary and the BEC($q$) then is a BEC with erasure probability $    s=\left(p-\frac{\epsilon}{1-q}\right)(1-q)+q\left (1-\left (p-\frac{\epsilon}{1-q} \right) \right)+\left (p-\frac{\epsilon}{1-q}\right)q=p+q-pq - \epsilon.$ Hence, $ C_f^{E}(p,q) \leq 1-s = (1-p)(1-q) + \epsilon.$ 

\textbf{Achievability :} The achievability scheme is essentially an ARQ type scheme. Alice simply transmits each of the $k$ bits in the message repeatedly until it is successfully received. If $e_\Lambda$ is the total number of erasures (a random quantity) that occur due to both the actions of Calvin and the BEC($q$), Alice needs $n = k + e_\Lambda$ channel uses for this scheme to succeed. Note that at channel use $t$, since Calvin does not know whether the BEC($q$) will introduce an erasure or not, 

we can show by the Chernoff bound that $P\left( e_{\Lambda} \leq ((p+q-pq)+\epsilon)n  \right )$ with probability at least $1-2^{\Omega(n\epsilon^2)}$ and hence, $ C_f^{E}(p,q) \geq (1-p)(1-q)-\epsilon$.

\end{proof}

In Fig. \ref{BECupperbounds}, we plot $C^E(p,q)$ and $C_f^E(p,q)$ as a function of $p$ for $q=0,0.3,0.6$.

\section{Results for Bit-Flips}\label{BFchannel}
In this section, we assume that Calvin can attempt to flip up to $pn$ bits and the random channel is a BSC($q$) instead of a BEC($q$). The input and output alphabets are $\mathcal{X} = \lbrace 0,1 \rbrace$ and $\mathcal{Y} = \lbrace 0,1\rbrace $. At time $k$, Calvin produces $a_k \in \mathcal{A}=\{0,1\}$ based on his side information which is the same as before, i.e., he knows $(x_1,x_2,\cdots, x_k)$, the codebook or the codeword distribution, and $(y_1,y_2,\cdots,y_{k-1})$. The received symbol at time $k$ at the receiver is $y_k = x_k + a_k +1$ with probability $q$ and $y_k = x_k + a_k $ with probability $1-q$ where $+$ denotes mod-$2$ addition and $q \in [0,1/2]$. The constraint on the adversary can be expressed as $weight(a_1,a_2,\cdots,a_n) \leq pn$. In contrast to the erasure case, note that a flip-attempt of Calvin can now be undone by the BSC. No feedback to the transmitter is assumed. For this model denoted BSC($q$)-ADV($p$)-FS, we prove an upper bound and use the result of \cite{chen2015characterization} to provide a simple achievable rate. The gap between the bounds gets larger when $q$ gets larger. Eliminating this gap and proving a tight capacity characterization is left as future work. 
\subsection{An Upper Bound $\overline{C}(p,q)$}
\begin{theorem}\label{flipconv_thm}
The capacity $C(p,q)$ of BSC($q$)-ADV($p$)-FS is bounded as $C(p,q)\leq \overline{C}(p,q)$ where
\begin{equation}\label{drloveresult}
    \overline{C}(p,q) = \min_{ \bar{p}: \bar{p} \in \mathcal{P}} \alpha(p,\bar{p},q) \left (  1 - h_2\left ( \frac{\bar{p}}{\alpha(p,\bar{p},q)} \star q\right)    \right),
\end{equation}
\begin{equation*}
    \alpha(p,\bar{p},q) = 1-4(p-\bar{p}) ~~,~~    \mathcal{P} = \left\lbrace \bar{p} : 0\leq \bar{p} \leq p  \right \rbrace
\end{equation*}
when $p<\frac{1}{4}$. When $p \geq \frac{1}{4}$, $C(p,q)=0$.
\end{theorem}
\begin{rem}
When $q=0,$ i.e., there is no BSC, the channel model reduces to that considered in \cite{dey2013upper}, and the capacity expression \eqref{drloveresult} matches with the result proved in \cite{dey2013upper}. 
\end{rem}

\begin{proof}
Fix a $\bar{p} \in [0,p]$. Suppose that for some $\epsilon >0 $, the transmitter attempts to communicate at a rate of $R = \alpha(p,\bar{p},q) \left (  1 - h_2\left ( \frac{\bar{p}}{\alpha(p,\bar{p},q)} \star q\right)    \right) +\epsilon$. We show that for sufficiently large $n$, under the proposed attack strategy for Calvin, the probability of decoding error in \eqref{maxerrorprob} is lower bounded by $\epsilon^{O(1/\epsilon)}$, a quantity \textit{independent} of $n$. Since the same argument works for any $\bar{p}$, the result in theorem \ref{flipconv_thm} holds.

Our proof is based on a \textit{babble and snoop, then push} attack inspired, in part from \cite{dey2013upper}. As noted before, the attack in \cite{dey2013upper} does not work if used directly. Let $\bx$ and $\by$ denote the transmitted and received words. Let $\bx_1=(x_1,x_2,\cdots,x_\ell)$ and $\bx_2=(x_{\ell+1},\cdots,x_n)$, where $\ell$ is specified shortly. Similarly, let $\by_1=(y_1,y_2,\cdots,y_\ell)$ and $\by_2=(y_{\ell+1},\cdots,y_n)$. The proposed attack consists of the following two phases:
\begin{itemize}
    \item \textbf{Babble and Snoop}: For the first $\ell=(\alpha(p,\bar{p},q)+\epsilon/2)n$ channel uses, Calvin injects random bit-flips and monitors Bob's reception - at channel use $i$, $1\leq i \leq \ell$, he flips bit $x_i$ with probability $\bar{p}n/\ell$. At the end of this phase, Calvin knows $\bx_1$ and $\by_1$.
    \item \textbf{Push}: Calvin samples a codeword $\bx'$ (corresponding to message $u'$) according to the conditional distribution $P_{\bX|\bY=\by_1}(.|\by_1)$. His goal is to confuse the receiver between $\bx$ and $\bx'$. At positions where $\bx_2$ and $\bx_2'$ agree, he does nothing. Positions $j$ where $\bx_2$ and $\bx_2'$ disagree, he flips $x_j$ with probability $1/2$. This way, the Bob cannot distinguish between $\bx$ and $\bx'$ or messages $u$ and $u'$ (even with the BSC($q$)) due to the fact that $p(\by_2 | \bx_2) = p(\by_2 | \bx_2')$. The proof relies on showing that with a small probability independent of $n$, $u$, $u'$ are distinct and $\bx_2$, $\bx_2'$ are sufficiently close.

\end{itemize}

Note that Calvin requires knowledge of $\bY_1$, i.e., the symbols received by Bob during the first phase of the attack. The presence of the BSC($q$) introduces additional equivocation at the receiver which Calvin is able to exploit to cause a reduction in rate. Here also, \textit{one-time block feedback} (of entire block $\by_1$) after the first $\ell$ channel uses is sufficient for the attack to succeed.

In the babble and snoop phase, by the Chernoff bound, Calvin uses at most $\bar{p}n+\epsilon n /64$ flips with probability at least $1-e^{-\Omega(\epsilon^2n)}$. Let this be denoted as event $E_1$. Conditioned on $E_1$, Calvin's remaining budget in the push phase is atleast $(p-\bar{p})n-\epsilon n / 64$. Define the set 
\begin{equation*}
   A_0 = \left \lbrace  \by_1 :  H(\bU \mid \bY_1 = \by_1) > \frac{n\epsilon}{4}  \right \rbrace. 
\end{equation*}

Defining the event $E_2=\lbrace \bY_1 \in A_0\rbrace$, we have the following lemma.
\begin{lemma}
     $P(E_2) \geq \epsilon /4$.
\end{lemma}
 \begin{proof}
 The proof closely follows claim 4 in \cite{dey2013upper}. Note that $\bU\to \bX_1 \to \bY_1$ is a markov chain and hence, by the data processing inequality and Calvin's actions in the babble phase,
\begin{equation*}
    I(\bU;\bY_1) \leq I(\bX_1;\bY_1) = \ell \left ( 1 - h_2 \left ( \frac{\bar{p}n}{\ell} \star q\right ) \right ). 
\end{equation*}
This is because the channel between $\bX_1$ and $\bY_1$ is now a cascade of $BSC(\bar{p}n/\ell)$ and $BSC(q)$.
Noting that $\ell = (\alpha + \epsilon/2 )n$, 
\begin{equation*}
    I(\bU;\bY_1) \leq   n (\alpha + \epsilon/2 ) \left ( 1 - h_2 \left ( \frac{\bar{p}}{\alpha + \epsilon /2} \star q\right )\right ). 
\end{equation*}
Since $I(\bU,\bY_1)=H(\bU) - H(\bU| \bY_1)$ and $H(\bU) = nR =  n\alpha\left (  1 - h_2\left ( \frac{\bar{p}}{\alpha} \star q\right)    \right) +n \epsilon $, we get,
\begin{multline*}
    H(\bU|\bY_1) \geq \frac{n\epsilon}{2} + \\ n \left ( (\alpha + \epsilon/2)h_2 \left ( \frac{\bar{p}}{\alpha + \epsilon /2} \star q\right ) - \alpha h_2\left(\frac{\bar{p}}{\alpha} \star q\right) \right ).
\end{multline*}
Now, the function $f(x) = xh_2\left ( \frac{\bar{p}}{x} \star q\right)$ is increasing in $x$, for any fixed $q\in(0,1/2)$. To see this, note that
\begin{equation*}
    \frac{ df}{dx} =h_2\left(\frac{\bar{p}}{x} \star q\right) + (2q-1)\frac{\bar{p}}{x} \log_2 \left (\frac{1-\frac{\bar{p}}{x}\star q}{\frac{\bar{p}}{x}\star q} \right ) >0
\end{equation*}
since $\frac{\bar{p}}{x}\star q < 1/2$ and $\log_2\left ( \frac{1-y}{y} \right) = \frac{d}{dy} h_2(y) >0 $ for $y\in(0,1/2)$. Hence, we have $H(\bU|\bY_1) = \mathbb{E}_{\bY_1} H(\bU|\bY_1=\by_1) \geq n\epsilon /2$. Finally, by Markov's inequality,
\begin{equation*}
    P\left ( nR - H(\bU|\bY_1 = \by_1) > nR - n\epsilon/4\right ) \leq 1 - \frac{\epsilon/4}{R - \epsilon/4}
\end{equation*}
which gives as desired,
\begin{equation*}
    P\left ( H(\bU \mid \bY_1 = \by_1) > \frac{n\epsilon}{4}  \right ) \geq \frac{\epsilon}{4}.
\end{equation*}
\end{proof}
Next, define the events $E_3=\lbrace \bU \neq \bU' \rbrace$ and $E_4 = \lbrace d_{H}(\bX_2,\bX_2') \leq 2(p-\bar{p})n - \epsilon n /8  \rbrace$. $E_3$ is the event that the message picked by the adversary to confuse Bob in the push phase is different from the one transmitted. Similarly, event $E_4$ ensures that Calvin's remaining flips are enough to carry his push attack. Using techniques from section A.2 of \cite{bassily2014causal} and claim 6 in \cite{dey2013upper}, we can now show the following. 
\begin{lemma} 
For $\by_1 \in A_0$,
   \begin{equation}\label{bscproofeventbd}
    P(E_3,E_4 \mid  \bY_1 = \by_1) \geq \frac{\epsilon}{48}\left (\frac{\epsilon}{5}\right )^{\frac{12}{\epsilon}-1} = \epsilon^{\mathcal{O}(1/\epsilon)}.
\end{equation}  
\end{lemma}
\begin{proof}
Consider sampling $t = \frac{12}{\epsilon}$ codewords $\mathcal{C}_t = \left \lbrace \bX^{(1)},\bX^{(2)},\cdots,\bX^{(t)} \right \rbrace$, each codeword sampled according to the conditional distribution $P_{\bX|\bY_1=\by_1}(.|\by_1)$. Let the messages corresponding to the codewords be $\bU_1,\bU_2,\cdots,\bU_t$ and let $E_5$ be the event that $\left \lbrace \bU_1, \bU_2, \cdots \bU_t \text{ are all distinct}\right \rbrace$ i.e. all of the codewords are distinct. We have from proposition 1, section A.2 from \cite{bassily2014causal} that for $\by_1 \in A_0$, for sufficiently large block length $n$,
\begin{equation*}
 P(E_5 \mid  \bY_1 = \by_1 ) \geq \left (\frac{\epsilon}{5}\right )^{t-1}.   
\end{equation*}
The average Hamming distance between the suffixes of codewords in $\mathcal{C}_t$ is defined as 
\begin{equation*}
    d_{avg}(\mathcal{C}_t) = \frac{1}{t(t-1)} \sum_{i\neq j} d_H\left ( \bX_2^{(i)} ,\bX_2^{(j)}\right ).
\end{equation*}
Recall that $\ell=(1-4(p-\bar{p})+\epsilon/2)n$. Conditioning on $E_5$, by Plotkin's bound we have
\begin{equation*}
       d_{avg}(\mathcal{C}_t) \leq \frac{1}{2}\frac{t}{t-1}(n-\ell) \leq 2(p-\bar{p})n - \epsilon n /6.
\end{equation*}
Thus for $\by_1 \in A_0$, we have
\begin{equation*}
    \mathbb{E}(d_{avg}(\mathcal{C}_t) \mid E_5, \bY_1 = \by_1 ) \leq 2(p-\bar{p})n - \epsilon n /6.
\end{equation*}
Now, since all of the $\bX^{(i)}$'s are picked independently, all pairs $(\bX^{(i)}, \bX^{(j)})$ have identical distribution. Thus,
\begin{multline*}
   \mathbb{E}(d_{avg}(\mathcal{C}_t) \mid E_5, \bY_1 = \by_1 ) =  \\ \mathbb{E}(d_H(\bX_2^{(1)},\bX_2^{(2)}) \mid E_5, \bY_1 = \by_1 )    
\end{multline*}
and also
\begin{multline*}
 \mathbb{E}(d_H(\bX_2^{(1)},\bX_2^{(2)}) \mid E_5, \bY_1 = \by_1 ) = \\ \mathbb{E}(d_H(\bX_2,\bX_2') \mid E_5, \bY_1 = \by_1 ).
\end{multline*}
Thus, we have
\begin{equation*}
    \mathbb{E}(d_H(\bX_2^{(1)},\bX_2^{(2)}) \mid E_5, \bY_1 = \by_1 ) \leq  2(p-\bar{p})n - \epsilon n /6.
\end{equation*}
and by Markov's inequality
\begin{multline}\label{flip_markov}
    P(d_H(\bX_2^{(1)},\bX_2^{(2)})>2(p-\bar{p})n - \epsilon n /8 \mid E_5, \bY_1 = \by_1 ) \leq \\ \frac{2(p-\bar{p})n - \epsilon n /6 }{2(p-\bar{p})n - \epsilon n /8 } = 1 - \frac{\epsilon}{48(p-\bar{p})-3\epsilon} \leq 1 - \frac{\epsilon}{48}.
\end{multline}
Thus, 
\begin{multline*}
    P(E_3,E_4 \mid \bY_1 = \by_1 ) = \\ P(d(\bX_2^{(1)},\bX_2^{(2)}) \leq 2(p-\bar{p})n - \epsilon n /8, \bU_1 \neq \bU_2 \mid  \bY_1 = \by_1 ) \\ \geq     P(d(\bX^{(1)}_2,\bX^{(2)}_2)  \leq 2(p-\bar{p})n - \epsilon n /8, E_5 \mid \bY_1 = \by_1 ).
\end{multline*}
where the last inequality holds because event $E_5$ is a subset of the event $\lbrace U_1\neq U_2 \rbrace$. We then have,
\begin{multline*}
     P(d(\bX^{(1)}_2,\bX^{(2)}_2)\leq 2(p-\bar{p})n - \epsilon n /8, E_5 \mid \bY_1 = \by_1 ) = \\
     P(d(\bX^{(1)}_2,\bX^{(2)}_2) \leq 2(p-\bar{p})n - \epsilon n /8 \mid E_5, \bY_1 = \by_1 ) \\ \times P(E_5\mid \bY_1 = \by_1  ).
\end{multline*}
From \eqref{flip_markov}, when $E_1$ occurs i.e. $\by_1 \in A_0$, we get,
\begin{equation*}
        P(E_3,E_4 \mid \bY_1 = \by_1 ) \geq \frac{\epsilon}{48}\left (\frac{\epsilon}{5}\right )^{\frac{12}{\epsilon}-1} = \epsilon^{\mathcal{O}(1/\epsilon)}
\end{equation*}
as we set out to prove. 
\end{proof}
 Now, in the push phase, Calvin injects $Ber(1/2)$ noise at $d_H(\bX_2,\bX_2^')$ positions. Conditioned on $E_1$, Calvin has at least a budget of $(p-\bar{p})n-\epsilon n /64$ bit-flips that remain. If $\ba_2$ is the error vector chosen by Calvin in the push phase, conditioned on $E_3$ and $E_4$ we have $\E  (d_H(\ba_2,\mathbf{0})) = (p-\bar{p})n - \epsilon n /16$. Further by the Chernoff bound, with probability at least $1-2^{-\Omega(\epsilon^2 n )}$, the distance $d_H(\ba_2,\mathbf{0})$ is within $3\epsilon n/ 64$ of its expected value. Let this event be $E_5$. Since $\E  (d_H(\ba_2,\mathbf{0})) + 3\epsilon n/64 = (p-\bar{p})n-\epsilon n /64$, the power constraint is respected w.h.p.. 

When events $E_1,E_3,E_4,E_5$ occur, the probability of decoding error is clearly at least $1/2$ since the receiver cannot distinguish between $\bx$ and $\bx'$. Since $P(E_1)\geq 1-e^{-\Omega(\epsilon^2 n)}$ and $P(E_5)\geq 1-e^{-\Omega(\epsilon^2 n)}$, the bound in \eqref{bscproofeventbd} together with the bound $P(E_2)\geq \epsilon/4$ implies for sufficiently large $n$, the maximum probability of error in \eqref{maxerrorprob} is at least of the order $\epsilon^{O(1/\epsilon)}$, a quantity independent of $n$ and the proof is complete.
\end{proof}

 The solution $\overline{C}(p,q)$ to the optimization problem in \eqref{drloveresult} has the following form:
\begin{itemize}
    \item $\overline{C}(p,q)>0$ for all $p\in[0,1/4)$ and $q \in [0,1/2)$. 
    
    \item For a fixed $q\in[0,1/2)$, there is a $p_0$ (a function of $q$) such that for $p\leq p_0$, $\overline{C}(p,q)$ is convex and equal to ($1-h_2(p \star q)$), which is the capacity when channels BSC($p$) and BSC($q$) are in cascade. Thus when $p\leq p_0$, the babble, snoop, and push strategy outlined here provides no benefit over a simpler adversarial strategy of injecting i.i.d. Ber($p$) bit-flips.
    
    \item It can be shown that the value of $p_0$ is the unique solution (different from $1/2$) of the equation
    \begin{equation}\label{Eqp_0}
        4 + (1+2q)\log_2\left(p_0\star q\right) + (3-2q)\log_2\left(1-p_0 \star q\right) = 0.
    \end{equation}
    
    \item For a fixed $q\in[0,1/2)$, $\overline{C}(p,q)$ for $p_0 \leq p\leq 1/4$ is a decreasing linear function in $p$ that intersects the p-axis at $p = 1/4$. Furthermore, $\overline{C}(p,q)$, $p_0\leq p\leq 1/4$ is in fact the \textit{tangent} to the curve $1-h_2(p \star q)$ at $p=p_0$. 
\end{itemize}
The proof of the above can be found in appendix \ref{solutionform}. In summary, for $q\in [0,1/2)$, we have
\begin{equation*}
    \overline{C}(p,q) = \begin{cases}
        1-h_2(p\star q) & 0\leq p \leq p_0 \\
    \frac{1-4p}{1-4p_0} \left ( 1 - h_2 \left(p_0\star q \right) \right ) & p_0\leq p \leq 1/4 \\
    0& p \geq 1/4
    \end{cases}
\end{equation*}
where $p_0$ is implicitly given by \eqref{Eqp_0}. In Fig. \ref{BSCupperboundsfig}, we plot $\overline{C}(p,q)$ as a function of $p$ for various values of $q$, specifically, $q=0.0,0.1,0.2$. 

\begin{figure}
    \centering
    \includegraphics[width=0.485\textwidth]{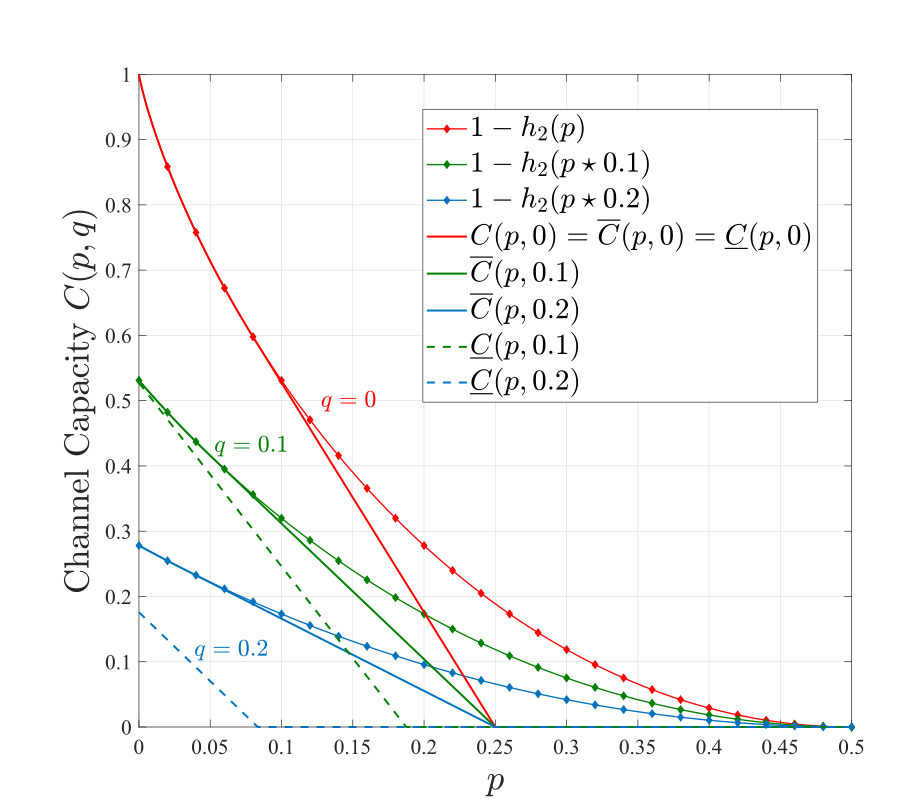}
    \caption{Upper bounds $\overline{C}(p,q)$ and lower bounds $\underline{C}(p,q)$ on the capacity of BSC($q$)-ADV($p$)-FS as a function of $p$. The cut-off value of $p$ beyond which $\overline{C}(p,q)=0$ is $p=1/4$ independent of $q$.}
    \label{BSCupperboundsfig}
\end{figure}

\subsection{An Achievable Rate $\underline{C}(p,q)$} 

\begin{theorem}
\label{achratenfb}
The capacity $C(p,q)$ of BSC($q$)-ADV($p$)-FS is at least $\underline{C}(p,q)=\overline{C}\left ((p \star q),0 \right)$.
\end{theorem}
\begin{proof}
From \cite{dey2013upper,chen2015characterization}, $\overline{C}(s,0)$ is a tight characterization of the capacity when there is no BSC present and Calvin has a total budget of $sn$ bit-flips. Since at channel use $k$, Calvin does not know if the BSC will cause a bit-flip, it can be shown by Chernoff bound that the total number of bit-flips is at most $\left((p \star q)+\epsilon\right)n$ with probability at least $1-e^{-\Omega(n\epsilon^2)}$. If we now assume that all of the $\left((p \star q)+\epsilon\right)n$ flips are chosen in an adversarial manner by Calvin, a rate of $\overline{C}((p \star q) + \epsilon,0)$ is achievable. 
\end{proof}
In Fig. \ref{BSCupperboundsfig}, we also plot achievable rates $\underline{C}(p,q)$ for $q=0,0.1,0.2$. As noted before, the gap between upper and lower bounds increases with $q$. 

\section{Conclusion}
In this work, we considered communicating over a stochastic channel (BEC/BSC) in the presence of a powerful adversary who can spy on both communicating terminals and inject further erasures/bit-flips at the input of the channel. For erasures, we gave a complete capacity characterization and for bit-flips, we proved interesting converse and achievability bounds. Future work includes characterizing capacity tightly for bit-flips and for the case where the adversary has no feedback snooping.

\bibliographystyle{IEEEtran}
\bibliography{refs}

\begin{thebibliography}{10}
\providecommand{\url}[1]{#1}
\csname url@samestyle\endcsname
\providecommand{\newblock}{\relax}
\providecommand{\bibinfo}[2]{#2}
\providecommand{\BIBentrySTDinterwordspacing}{\spaceskip=0pt\relax}
\providecommand{\BIBentryALTinterwordstretchfactor}{4}
\providecommand{\BIBentryALTinterwordspacing}{\spaceskip=\fontdimen2\font plus
\BIBentryALTinterwordstretchfactor\fontdimen3\font minus
  \fontdimen4\font\relax}
\providecommand{\BIBforeignlanguage}[2]{{%
\expandafter\ifx\csname l@#1\endcsname\relax
\typeout{** WARNING: IEEEtran.bst: No hyphenation pattern has been}%
\typeout{** loaded for the language `#1'. Using the pattern for}%
\typeout{** the default language instead.}%
\else
\language=\csname l@#1\endcsname
\fi
#2}}
\providecommand{\BIBdecl}{\relax}
\BIBdecl

\bibitem{langberg2009binary}
M.~Langberg, S.~Jaggi, and B.~K. Dey, ``Binary causal-adversary channels,'' in
  \emph{2009 IEEE International Symposium on Information Theory}.\hskip 1em
  plus 0.5em minus 0.4em\relax IEEE, 2009, pp. 2723--2727.

\bibitem{code_online_adv_old}
B.~K. {Dey}, S.~{Jaggi}, and M.~{Langberg}, ``Codes against online
  adversaries,'' in \emph{2009 47th Annual Allerton Conference on
  Communication, Control, and Computing (Allerton)}, 2009, pp. 1169--1176.

\bibitem{improved_upbounds}
B.~K. {Dey}, S.~{Jaggi}, M.~{Langberg}, and A.~D. {Sarwate}, ``Improved upper
  bounds on the capacity of binary channels with causal adversaries,'' in
  \emph{2012 IEEE International Symposium on Information Theory Proceedings},
  2012, pp. 681--685.

\bibitem{dey2013upper}
B.~K. Dey, S.~Jaggi, M.~Langberg, and A.~D. Sarwate, ``Upper bounds on the
  capacity of binary channels with causal adversaries,'' \emph{IEEE
  Transactions on Information Theory}, vol.~59, no.~6, pp. 3753--3763, 2013.

\bibitem{bassily2014causal}
R.~Bassily and A.~Smith, ``Causal erasure channels,'' in \emph{Proceedings of
  the twenty-fifth annual ACM-SIAM symposium on Discrete algorithms}.\hskip 1em
  plus 0.5em minus 0.4em\relax SIAM, 2014, pp. 1844--1857.

\bibitem{chen2015characterization}
Z.~Chen, S.~Jaggi, and M.~Langberg, ``A characterization of the capacity of
  online (causal) binary channels,'' in \emph{Proceedings of the forty-seventh
  annual ACM symposium on Theory of computing}, 2015, pp. 287--296.

\bibitem{error_erasure}
Z.~{Chen}, S.~{Jaggi}, and M.~{Langberg}, ``The capacity of online (causal) $q$
  -ary error-erasure channels,'' \emph{IEEE Transactions on Information
  Theory}, vol.~65, no.~6, pp. 3384--3411, 2019.

\bibitem{online_largea}
B.~K. {Dey}, S.~{Jaggi}, and M.~{Langberg}, ``Codes against online adversaries:
  Large alphabets,'' \emph{IEEE Transactions on Information Theory}, vol.~59,
  no.~6, pp. 3304--3316, 2013.

\bibitem{sufficientlymyopic}
B.~K. Dey, S.~Jaggi, and M.~Langberg, ``Sufficiently myopic adversaries are
  blind,'' \emph{IEEE Transactions on Information Theory}, vol.~65, no.~9, pp.
  5718--5736, 2019.

\bibitem{dey2010codingdelayed}
B.~K. Dey, S.~Jaggi, M.~Langberg, and A.~D. Sarwate, ``Coding against delayed
  adversaries,'' in \emph{2010 IEEE International Symposium on Information
  Theory}.\hskip 1em plus 0.5em minus 0.4em\relax IEEE, 2010, pp. 285--289.

\bibitem{dey2016bitofdelay}
------, ``A bit of delay is sufficient and stochastic encoding is necessary to
  overcome online adversarial erasures,'' in \emph{2016 IEEE International
  Symposium on Information Theory (ISIT)}.\hskip 1em plus 0.5em minus
  0.4em\relax IEEE, 2016, pp. 880--884.

\bibitem{dey2019interplay}
B.~K. Dey, S.~Jaggi, M.~Langberg, A.~D. Sarwate, and C.~Wang, ``The interplay
  of causality and myopia in adversarial channel models,'' in \emph{2019 IEEE
  International Symposium on Information Theory (ISIT)}.\hskip 1em plus 0.5em
  minus 0.4em\relax IEEE, 2019, pp. 1002--1006.

\bibitem{zhang2018covert}
Q.~E. Zhang, M.~Bakshi, and S.~Jaggi, ``Covert communication over adversarially
  jammed channels,'' in \emph{2018 IEEE Information Theory Workshop
  (ITW)}.\hskip 1em plus 0.5em minus 0.4em\relax IEEE, 2018, pp. 1--5.

\bibitem{amita1}
A.~J. {Budkuley} and S.~{Jaggi}, ``Communication over an arbitrarily varying
  channel under a state-myopic encoder,'' in \emph{2018 IEEE International
  Symposium on Information Theory (ISIT)}, 2018, pp. 616--620.

\bibitem{berlekamp1964block}
E.~R. Berlekamp, ``Block coding with noiseless feedback,'' Ph.D. dissertation,
  Massachusetts Institute of Technology, 1964.

\bibitem{lebedev2016coding}
V.~S. Lebedev, ``Coding with noiseless feedback,'' \emph{Problems of
  Information Transmission}, vol.~52, no.~2, pp. 103--113, 2016.

\bibitem{zigangirov1976number}
K.~Zigangirov, ``On the number of correctable errors for transmission over a
  binary symmetrical channel with feedback,'' \emph{Problemy Peredachi
  Informatsii}, vol.~12, no.~2, pp. 3--19, 1976.

\bibitem{csiszarkornerbook}
I.~Csiszar and J.~K{\"o}rner, \emph{Information theory: coding theorems for
  discrete memoryless systems}.\hskip 1em plus 0.5em minus 0.4em\relax
  Cambridge University Press, 2011.

\bibitem{surveyAVC}
A.~{Lapidoth} and P.~{Narayan}, ``Reliable communication under channel
  uncertainty,'' \emph{IEEE Transactions on Information Theory}, vol.~44,
  no.~6, pp. 2148--2177, 1998.

\end{thebibliography}

\appendices
\section{Proof of Theorem \ref{thm_er} : Achievability}
We provide the details of proof for the achievability of $C^E(p,q)$. 

\textbf{Random Code Distribution}: Alice is endowed with a set of private keys for encoding, $\mathcal{S}=\lbrace 1,2,\cdots,2^{nS}\rbrace$. The encoding procedure is carried out in chunks, each of size $n\theta$ where $\theta<1$ is a quantization parameter set to $\theta=\frac{\epsilon}{4}$. We also set $S=\frac{\theta^3}{8}$. Let $\Gamma$ be the uniform distribution over stochastic codes $\mathcal{C}: \mathcal{U} \times \mathcal{S} \to \mathcal{X}^{n\theta}$. Then each chunk $i$, $1\leq i \leq \frac{1}{\theta}$, is associated to a stochastic code $\mathcal{C}_i$ drawn independently from the distribution $\Gamma$.

\textbf{Encoding:} For message $u\in \mathcal{U}$ and keys $s_1,s_2,\cdots,s_{\frac{1}{\theta}}$, the codeword $\bx$ selected for transmission is 
\begin{equation*}
    \bx = \mathcal{C}_1(u,s_1)\circ\mathcal{C}_2(u,s_2)\circ \cdots \mathcal{C}_{\frac{1}{\theta}}(u,s_{\frac{1}{\theta}}),
\end{equation*}
where $\circ$ represents the concatenation operator. We refer to codeword $\mathcal{C}_i(u,s_i)$ as the $i^{th}$ sub-codeword or the $i^{th}$ chunk and the code $\mathcal{C}_i$ as the $i^{th}$ sub-code. Each secret or key $s_i$ for encoding with $\mathcal{C}_i$ is chosen uniformly randomly from $\mathcal{S}$.

We define the set $\mathcal{T}=\lbrace n\theta,2n\theta,\cdots, n-n\theta\rbrace$ containing indices of the chunk ends. For some $t\in\mathcal{T}$ where $t=kn\theta$, we refer to $ \mathcal{C}_1\circ\mathcal{C}_2\circ \cdots \mathcal{C}_{k}$ as the left mega sub-code w.r.t. $t$ and $ \mathcal{C}_{k+1}\circ\mathcal{C}_2\circ \cdots \mathcal{C}_{\frac{1}{\theta}}$ as the right mega sub-code w.r.t. $t$. Accordingly, the concatenation of the first $k$ sub-codewords is be referred to as the left mega sub-codeword w.r.t. $t$, and that of the last $\frac{1}{\theta}-k$ sub-codewords is referred to as right mega sub-codeword w.r.t. $t$. We shall also denote the key sequences used to encode the left and the right mega-subcodewords as $s_{left}=(s_1,s_2,\cdots,s_k)$ and $s_{right}=(s_{k+1},s_{k+2},\cdots,s_{\frac{1}{\theta}})$.

\textbf{Estimation Bounds}: We review some simple bounds we use later to show that decoding succeeds w.h.p.. Recall from before that $\lambda^{a}_t$ denotes the number of erasures added by Calvin up until time $t$. By the Chernoff bound we have then that, for $\delta>0$ to be set later, the total number of erasures that Bob observes at time $t$ satisfies $$\lambda_t \in \left [\lambda^{a}_t + (t-\lambda^{a}_t)(q-\delta),\lambda^{a}_t + (t-\lambda^{a}_t)(q+\delta)  \right] $$ with probability at least $P_\delta = 1-2^{\Omega(\epsilon^2 n)}$. Thus, $\hat{\lambda}_t = \lambda_t-qt$, Bob's estimate of the number of adversarial erasures that do not coincide with BEC($q$) erasures satisfies w.h.p. $$\hat{\lambda}_t \in \left[  \lambda^a_t(1-q+\delta) - \delta t ,  \lambda^a_t(1-q-\delta) + \delta t \right].$$
Recall that with the modified list-decoding and list refinement conditions, Bob selects the smallest value of $t^*$ satisfying
\begin{equation} \label{eq:phase1_cond_m_2}
{\lambda}_{t^*} - qt^*  \leq t^*(1-q)(1-\theta) - Rn 
\end{equation}
and 
\begin{equation} \label{eq:phase2_cond_m_2}
np(1-q) - ({\lambda}_{t^*}- qt^*)  \leq \frac{(n-t^*)(1-q)(1-\theta)}{2}.
\end{equation}
We further restrict the choice of $t^*$ so that we must have $t^* \in \mathcal{T}$ i.e. $t^*$ must correspond to a chunk end. Now, from the preceding discussion, w.h.p. we have that 
\begin{multline}\label{eq:hatlambda_range}
    {\lambda}_{t^*} - qt^* =  \hat{\lambda}_{t^*}  \in  [  \lambda^a_{t^*}(1-q+\delta) - \delta t^* , \\ \lambda^a_{t^*}(1-q-\delta) + \delta t^* ].
\end{multline}
Let $\mathcal{Z} = \left[  \lambda^a_{t^*}(1-q+\delta) - \delta t^* ,  \lambda^a_{t^*}(1-q-\delta) + \delta t^* \right]$. By a similar analysis as in {\cite[Claim B.3]{chen2015characterization}}, we can show existence of $t^*$ that satisfies both \eqref{eq:phase1_cond_m_2} and \eqref{eq:phase2_cond_m_2} for any realization of $\hat{\lambda}_{t^*} \in \mathcal{Z}$.

\textbf{Calvin's Unused Budget}: We now prove an upper bound on the number of adversarial erasures that Calvin is left with to add on to the right mega sub-codeword. Since the total budget is $pn$, the remaining number erasures is $pn-\lambda^a_{t^*}$. From \eqref{eq:phase2_cond_m_2} and \eqref{eq:hatlambda_range}, for any $\hat{\lambda}_{t^*} \in \mathcal{Z}$, we have
\begin{equation*}
  pn-\lambda^a_{t^*} \leq \frac{(n-t^*)(1-\theta)}{2} + \frac{\delta(t-\lambda^a_{t^*})}{1-q}.
\end{equation*}
Choosing for instance $\delta = \frac{1}{16}(1-q)\theta^2$, we can show that
\begin{equation}\label{eq:Remaining_flips}
  pn-\lambda^a_{t^*} \leq (n-t^*)\left ( \frac{1}{2}-\frac{7\theta}{16}\right ).
\end{equation}

\textbf{List Decoding}: From {\cite[Claims B.5-B.7]{chen2015characterization}}, for sufficiently large $n$, we have that with probability at least $P_\delta\left ( 1 - \frac{1}{n}\right ) \geq \frac{1}{2}\left ( 1 - \frac{1}{n}\right )$, the size of the list of messages $\mathcal{L}$ obtained by Bob in the list-decoding phase is at most a constant, i.e. $|\mathcal{L}| < C/\epsilon$ for some constant C.

\textbf{List Refinement}: For some chunk end $t\in \mathcal{T}$ where $t=kn\theta$, $\by_1=(y_1,y_2,\cdots,y_t)$ and $\by_2=(y_{t+1},\cdots,y_n)$ are the left mega received word and the right mega received word w.r.t. $t$ respectively. Consider the list of messages $\mathcal{L}$ obtained by Bob by list-decoding the left mega received word $\by_1$. Let $u^*$ be the true message chosen by Alice for transmission and let $\mathcal{L}(u^*)$ be the set of all possible right mega sub-codewords w.r.t $t$ for each message in $\mathcal{L} \setminus \lbrace u^*\rbrace$ i.e.
\begin{multline*}
    \mathcal{L}(u^*) = \lbrace \mathcal{C}_{k+1}(u,s_{k+1})\circ\mathcal{C}_{k+2}(u,s_{k+2})\circ \cdots \mathcal{C}_{\frac{1}{\theta}}(u,s_{\frac{1}{\theta}}) : \\ u \in \mathcal{L} , u \neq u^*, (s_{k+1},\cdots,s_{1/\theta})\in \mathcal{S}^{\frac{1}{\theta}-k}   \rbrace.
\end{multline*}
For notational convenience, also enumerate $\mathcal{L}(u^*)$ containing codewords of length $(n-t)$ as $\mathcal{L}(u^*)=\lbrace \bw_1, \bw_2,\cdots,\bw_{|\mathcal{L}(u^*)|} \rbrace$. The right mega-subcodeword for the true message is 
\begin{multline*}
 \bx_2(s_{right},u^*) = \\ \mathcal{C}_{k+1}(u^*,s_{k+1})\circ\mathcal{C}_{k+2}(u^*,s_{k+2})\circ \cdots \mathcal{C}_{\frac{1}{\theta}}(u^*,s_{\frac{1}{\theta}})
\end{multline*} which we emphasize is a function of the specific realization of $s_{right=}(s_{k+1},\cdots,s_{\frac{1}{\theta}})$ during encoding.

We would like our code design to satisfy the following distance condition
\begin{equation}\label{eq:dist_bound}
    d_H\left ( \bx_2(s_{right},u^*), \bw_j  \right) \geq (n-t)\left ( \frac{1}{2}-\frac{3\theta}{8}\right ) ~ ~\forall \bw_j \in \mathcal{L}(u^*).
\end{equation}
Equation \eqref{eq:dist_bound} is a key property that guarantees successful decoding. It ensures that the right mega sub-codeword for the transmitted message is sufficiently far in Hamming distance from the right mega sub words for any of the other messages in list $\mathcal{L}$. We show that \eqref{eq:dist_bound} indeed occurs w.h.p., for almost all possible sequence of secrets $s_{right}$.

\begin{lemma}({Modified from \cite[Claims B.11-B.14]{chen2015characterization}})
For sufficiently large $n$, with probability at least $1-2^{-n}$, a code drawn from the random ensemble satisfies the following property : for every chunk end $t \in \mathcal{T}$, for every message $u^*$, and every list $\mathcal{L}$ of size at most $O(1/\epsilon)$, we have that \eqref{eq:dist_bound} holds for at least a $(1-2^{-nS/4})$ portion of all possible secret sequences $s_{right}$.
\end{lemma}
\begin{proof}
Given a sequence of secrets $s_{right}=(s_{k+1},\cdots,s_{\frac{1}{\theta}})$, message $u^*$ and list $\mathcal{L}$, we first show that \eqref{eq:dist_bound} holds w.h.p.. Let radius $r = \left( \frac{1}{2}-\frac{3\theta}{8}\right ) $. We surround each word $\bw_j \in \mathcal{L}(u^*)$ with a Hamming ball of radius $r$ and the union of all the balls is the so called forbidden region.

For \eqref{eq:dist_bound} to hold, we must have that $\bx_2(s_{right},u^*)$ is outside all these balls, i.e. outside the forbidden region. Due to the code construction, $\bx_2(s_{right},u^*)$ is uniformly distributed over all possible binary vectors of length $(n-t)$ and thus it is enough to bound the size of the forbidden region. If the size of the list $\mathcal{L}$ is $L$, the size of $\mathcal{L}(u^*)$ is at most $L.2^{nS\left( \frac{1}{\theta} - \frac{t}{n\theta}\right)}$. Hence the number of codewords in the forbidden region is at most
\begin{equation*}
    L.2^{nS\left( \frac{1}{\theta} - \frac{t}{n\theta}\right)} \sum_{j=0}^r \binom{n-t}{j} < 2^{(n-t)\left( \frac{\log_2 L}{n-t} + \frac{S}{\theta} + h_2\left( \frac{1}{2} - \frac{3\theta}{8}\right) \right)}.
\end{equation*}
From the Taylor expansion of function $h_2(x)$ in a neighborhood of $1/2$, we have 
\begin{equation*}
\begin{split}
    h_2\left( \frac{1}{2} - \frac{3\theta}{8}\right) &< 1 - \frac{1}{2\ln(2)} \left( 1 - 2\left( \frac{1}{2}- \frac{3\theta}{8}\right)  \right)^2 \\
    &= 1 - \frac{9\theta^2}{32\ln(2)}
\end{split}
\end{equation*}
Let $\eta = \frac{\theta^2}{4}$. For sufficiently large $n$, we have 
\begin{multline*}
    \left( \frac{\log_2 L}{n-t} + \frac{S}{\theta} + h_2\left( \frac{1}{2} - \frac{3\theta}{8}\right) \right) < \\  \left( \frac{\log_2 L}{n-t} + \frac{S}{\theta} + \left(1 - \frac{9\theta^2}{32\ln(2)}\right) \right) <1-\eta.
    \end{multline*}
Hence, the total number of codewords in the forbidden region is at most $2^{(n-t)(1-\eta)}$ and we have
\begin{multline*}
    P\Big(\bx_2(s_{right},u^*) \text{ is outside the forbidden region} \Big) > \\ \frac{2^{(n-t)} - 2^{-(n-t)(1-\eta)}}{2^{n-t}} = 1 - 2^{-(n-t)\eta}.
\end{multline*}
From here on, the rest of the steps in the proof follow exactly the analysis of Claims B.12-B.14 in \cite{chen2015characterization}.
\end{proof}

\textbf{Success of Unique Decoding}: From the preceding discussion, there exists a code in our random ensemble that satisfies the following simultaneously:
\begin{itemize}
    \item For $t^*$ satisfying \eqref{eq:phase1_cond_m_2} and \eqref{eq:phase2_cond_m_2}, the size of the list $\mathcal{L}$ obtained during list decoding is at most $C/\epsilon$ for some constant $C$. Further, the transmitted message $u^*$ is inside list $\mathcal{L}$.
    \item For almost all possible realizations of secret sequences $s_{right}$, the right mega codeword corresponding to message $u^*$ denoted $\bx_2(s_{right},u^*)$, is at least $(n-t^*)\left ( \frac{1}{2}-\frac{3\theta}{8}\right )$ away from any codeword in the set $\mathcal{L}(u^*)$.
\end{itemize}

Recall from \eqref{eq:Remaining_flips} that with probability at least $1-2^{-\Omega(\epsilon^2 n)}$, Calvin has at most $pn-\lambda^a_{t^*} \leq (n-t^*)\left ( \frac{1}{2}-\frac{7\theta}{16}\right )$ erasures that remain. Now, consider any arbitrary codeword $\bw_j \in \mathcal{L}(u^*)$ that is associated with message $u' \neq u^*$. Clearly, if Calvin wishes to confuse Bob between $u'$ and $u^*$, the best strategy is to add all erasures at positions where $\bw_j$ and $\bx_2(s_{right},u^*)$ disagree. However, this still leaves at least $(n-t^*)\left ( \frac{1}{2}-\frac{3\theta}{8} \right ) -(n-t^*)\left ( \frac{1}{2}-\frac{7\theta}{16}\right ) =(n-t^*) \frac{\theta}{16} $ positions where $\bw_j$ and $\bx_2(s_{right},u^*)$ disagree but no adversarial erasures are added. 

The only way that Bob is unable to distinguish between $\bw_j$ and $\bx_2(s_{right},u^*)$ and hence makes a decoding error is when the BEC($q$) erases all of the $(n-t^*) \frac{\theta}{16} $ bits from $\bx_2(s_{right},u^*)$ that Calvin could not erase. However, by the Chernoff bound, this event occurs with probability at most $2^{-\Omega(n\theta^2)}=2^{-\Omega(n\epsilon^2)}$. By a union bound over all $\bw_j \in \mathcal{L}(u^*)$, we have that a decoding error occurs with exponentially small probability. Thus, Bob succeeds in determining the transmitted message $u^*$ and the proof is complete.

 \section{Form of $\overline{C}(p,q)$} \label{solutionform}
 Fix a $q\in[0,1/2)$. The optimization problem \eqref{drloveresult} in Theorem \ref{flipconv_thm} is
 \begin{equation}\label{optprobKq}
  \min_{0\leq x \leq p} f(x) 
\end{equation}
where 
\begin{equation*}
    f(x) = \left(1-4p+4x\right) \left (  1 - h_2\left ( \frac{x}{1-4p+4x} \star q \right)    \right). 
\end{equation*}
When $p=1/4$, $f(x)=0$ at $x=0$ and hence  $\overline{C}(p,q)=0$ when $p=1/4$. Differentiating the objective function in \eqref{optprobKq},
\begin{equation*}
      \frac{d }{dx} \left ( \left(1-4p+4x\right) \left (  1 - h_2\left ( \frac{x}{1-4p+4x} \star q \right) \right) \right) = 0
\end{equation*}
we get,
\begin{multline*}
     4 + (2q+1)\log_2 \left ( \frac{x(1+2q)+q(1-4p)}{1-4p+4x}\right ) +    (3- \\ 2q)  \log_2 \left ( \frac{1-4p+4x-x(1+2q)-q(1-4p)}{1-4p+4x}\right ) = 0.
\end{multline*}
Solution $x^*$ has the form $x^*=\frac{1-4p}{\alpha-3}$ where $\alpha$ satisfies
\begin{multline*}
    4 + (1+2q)\log_2 \left ( \frac{1-q(1-\alpha)}{1+\alpha} \right )+ \\
    (3-2q)\log_2 \left( \frac{\alpha+q(1-\alpha)}{1+\alpha} \right)= 0. 
\end{multline*}
Since $0\leq x \leq p$, we must have $\frac{1-4p}{\alpha-3} \leq p \implies p \geq \frac{1}{1+\alpha} = p_0$. Thus, for $p \in [p_0,1/4]$, the minimizer in \eqref{optprobKq} is $x^*= \frac{(1-4p)p_0}{1-4p_0}$ where $p_0$ satisfies 
\begin{equation}\label{p_0cond}
        4 + (1+2q)\log_2\left(p_0\star q\right) + (3-2q)\log_2\left(1-p_0 \star q\right) = 0,
    \end{equation}
and the capacity upper bound becomes 
\begin{equation*}
\begin{split}
    \overline{C}(p,q) & = \frac{1-4p}{1-4p_0}\left( 1-h_2\left (  \frac{\frac{p_0}{1-4p_0}}{1 + \frac{p_0}{1-4p_0}} \star q\right )\right) \\
    & = \frac{1-4p}{1-4p_0}\left( 1-h_2(p_0 \star q)\right) 
\end{split}.
\end{equation*}
Thus, $ \overline{C}(p,q)$, $p_0\leq p \leq 1/4$ is a straight line that intersects the $p$-axis at $p=1/4$. For $p\in[0,p_0]$, the minimizer in \eqref{optprobKq} is $x^*=p$ and the expression is
\begin{equation*}
     \overline{C}(p,q) =  1-h_2(p\star q).
\end{equation*}
Next we show that, $ \overline{C}(p,q)$, $p_0\leq p \leq 1/4$ is in fact the tangent to the curve $1-h_2(p\star q)$ at $p=p_0$. Consider the line $L(p)$ that is tangent to $1-h_2(p\star q)$ and passes through $(1/4,0)$. Its equation can be written as $L(p) = C(1-4p)$ where $C$ is a constant. Suppose that $L(x)$ intersects $1-h_2(p \star q)$ at $p=\tilde{p_0}$. To complete the proof, it suffices to show that $\tilde{p_0} = p_0$ i.e. $\tilde{p_0}$ satisfies \eqref{p_0cond}. Since $L(p)$ is the tangent to $1-h_2(p\star q)$ at $p=\tilde{p_0}$, we have
\begin{equation*}
    \frac{d}{dp}L(p)\Big|_{p=\tilde{p_0}} = \frac{d}{dp}(1-h_2(p\star q))\Big|_{p=\tilde{p_0}}
\end{equation*}
which gives
\begin{equation}\label{tanproof1}
    -4C = (1-2q)\log_2 \left (\frac{\tilde{p_0}\star q}{1-\tilde{p_0}\star q} \right).
\end{equation}
We also have 
\begin{equation}\label{tanproof2}
    L(\tilde{p_0}) = 1-h_2(\tilde{p_0}\star q)=C(1-4\tilde{p_0}).
\end{equation}
Eliminating the constant $C$ from \eqref{tanproof1} and \eqref{tanproof2}, $\tilde{p_0}$ satisfies the equation
\begin{equation*}
    (1-2q)\log_2 \left (\frac{\tilde{p_0}\star q}{1-\tilde{p_0}\star q} \right) = -4\left (\frac{1-h_2(\tilde{p_0}\star q)}{1-4\tilde{p_0}}\right ).
\end{equation*}
Rearranging the terms,
\begin{multline*}
    \left ( (1-2q)(4\tilde{p_0}-1) -(\tilde{p_0}\star q)\right)\log_2(\tilde{p_0}\star q) - \\   \left ((1-2q)(4\tilde{p_0}-1) - (1-(\tilde{p_0}\star q))\right )\log_2(1 - \tilde{p_0}\star q) = 4
\end{multline*}
which simplifies to
\begin{equation*}
 4 + (1+2q)\log_2\left(\tilde{p_0}\star q\right) + (3-2q)\log_2\left(1-\tilde{p_0} \star q\right) = 0 
\end{equation*}
which is the same as \eqref{p_0cond}. Hence, $p_0=\tilde{p_0}$ and the claim holds.

\end{document}